\documentclass[aps,pra,reprint,groupaddress,superscriptaddress, twocolumn]{revtex4-1}

\usepackage{algpseudocode}

\usepackage{algorithm}
\usepackage{enumerate}

\usepackage[utf8]{inputenc}
\usepackage{amssymb}
\usepackage{graphicx}
\usepackage{amsmath}
\usepackage{amsbsy}

\usepackage{amsthm}
\newtheorem{theorem}{Theorem}
\newtheorem{lemma}{Lemma}
\newtheorem{proposition}{Proposition}

\theoremstyle{definition}
\newtheorem{definition}{Definition}

\usepackage{bbm}
\usepackage{bm}

\DeclareMathOperator{\poly}{poly}

\usepackage[
bookmarksopen,
bookmarksdepth=2,
breaklinks=true, colorlinks
]{hyperref}

\expandafter\def\expandafter\UrlBreaks\expandafter{\UrlBreaks%  save the current one
  \do\a\do\b\do\c\do\d\do\e\do\f\do\g\do\h\do\i\do\j%
  \do\k\do\l\do\m\do\n\do\o\do\p\do\q\do\r\do\s\do\t%
  \do\u\do\v\do\w\do\x\do\y\do\z\do\A\do\B\do\C\do\D%
  \do\E\do\F\do\G\do\H\do\I\do\J\do\K\do\L\do\M\do\N%
  \do\O\do\P\do\Q\do\R\do\S\do\T\do\U\do\V\do\W\do\X%
  \do\Y\do\Z}

\usepackage{color}

\def\EQ#1{\begin{eqnarray}#1\end{eqnarray}}

\def\ket#1{\left| #1 \right\rangle}
\newcommand\mydots{\hbox to 1em{.\hss.\hss.}}

\DeclareMathOperator{\enc}{Enc}

\DeclareMathOperator{\pos}{pos}
\DeclareMathOperator{\len}{len}

\DeclareMathOperator{\EffEnc}{EffEnc}
\def\ch{\textsf{Check}}
\def\extract{\textsf{Extract}}
\def\select{\textsf{Select}}
\def\union{\textsf{Union}}
\def\contains{\textsf{Contains}}
\def\shift{\textsf{Shift}}
\def\Min{\textsf{Min}}
\def\append{\textsf{Append}}

\def\calculate{\textsf{Calculate}}
\def\merge{\textsf{Merge}}

\def\qball{\textsc{QBall}}
\def\promiseball{\textsc{PromiseBall}}
\def\fastball{\textsc{FastBall}}

\newcommand\rcall{r_{\text{call}}}

\begin{document}
\title{Computational speedups using small quantum devices}
\date{\today}
\author{Vedran Dunjko}
\email{v.dunjko@liacs.leidenuniv.nl}
\address{Max Planck Institut f\"{u}r Quantenoptik, Hans-Kopfermann-Str. 1,
85748 Garching, Germany}
\address{LIACS, Leiden University, Niels Bohrweg 1, 2333 CA Leiden, The Netherlands}
\author{Yimin Ge}
\email{yimin.ge@mpq.mpg.de}
\address{Max Planck Institut f\"{u}r Quantenoptik, Hans-Kopfermann-Str. 1,
85748 Garching, Germany}
\author{J. Ignacio Cirac}
\email{ignacio.cirac@mpq.mpg.de}
\address{Max Planck Institut f\"{u}r Quantenoptik, Hans-Kopfermann-Str. 1,
85748 Garching, Germany}

\begin{abstract}

Suppose we have a small quantum computer with only $M$ qubits. Can such a device genuinely speed up certain algorithms, even when the problem size is much larger than $M$? Here we answer this question to the affirmative. We present a hybrid quantum-classical algorithm to solve 3SAT problems involving $n\gg M$ variables that significantly speeds up its fully classical counterpart. This question may be relevant in view of the current quest to build small quantum computers.
\end{abstract}

\maketitle

Quantum computers can use the superposition principle to speed up certain computations. However, it is not clear if they can be useful when they are, as expected for the foreseeable future, limited in size. The reason is that classical and quantum algorithms alike typically exploit global structures of problems, and restricting superpositions to certain block sizes will break that structure. 
Thus, one may expect that for problems where arbitrarily sized quantum computers offer advantages, small quantum computers may end up being of no significant help given large inputs. 

In this paper we study this problem and show that this is not generally true: there are relevant algorithms which utilize the global structure, and where quantum computers significantly smaller than the problem size can offer significant speedups. More precisely, we focus on the famous algorithm of Sch\"{o}ning for boolean satisfiability, and present a modified hybrid quantum-classical algorithm which significantly outperforms its purely classical version even given small quantum computers.

Satisfiability (SAT) problems are among the basic computational problems and naturally appear in many contexts dealing with combinatorial optimization, like scheduling or planning tasks, as well as in statistical physics.
A prominent SAT problem is  
3SAT, which involves clauses with up to three literals. 3SAT is the canonical example of  so-called NP-complete problems, which are believed to be exponentially difficult even for quantum computers. Nevertheless, quantum computers can still accelerate their solution \cite{2004_Ambainis,2018_Ambainis}, and given their ubiquity, they may become one of the most important applications of quantum computers.
However, the best quantum algorithms, which essentially ``quantum-enhance'' classical SAT solvers \cite{2004_Ambainis}, require a large number of qubits, and are not directly applicable given a quantum computer of limited size.
There are a number of possibilities how one could use a limited-size quantum device. For instance, one could speed up smaller structure-independent subroutines which occur, \emph{e.g.,} in the preparation phases of algorithms (\emph{e.g.,} whenever search over few items is performed, one could utilize Grover's algorithm \cite{1996_Grover}). 
However, for obtaining ``genuine'' speedups, \emph{e.g,} those which interpolate between the runtimes of the fully classical and a fully quantum algorithm, according to the size of the available quantum device, one should attack the actual computational bottlenecks. 
As we will show, if this is done in a straightforward fashion, one may encounter a {\em threshold effect}: if the quantum device is too small, \emph{i.e.}, can handle only a small fraction of the entire instance,
a na\"{i}ve hybrid algorithm turns out to be \emph{slower} than its classical version.

Here we demonstrate how this effect can sometimes be circumvented, in the context of satisfiability problems.
Specifically, we provide a quantum-assisted version of a well-understood classical algorithm, which achieves genuine improvements given quantum computers of basically any size, avoiding the aforementioned threshold effect \footnote{Our scenario is related to, but distinct from, the work in \cite{2016_Bravyi}. There, the authors consider how a classical computer can help simulate a given quantum algorithm which requires only slightly more qubits than are available. In contrast, we change the algorithm, \emph{i.e.},  propose hybrid algorithms which use a significantly smaller quantum device, to speed up a fixed classical algorithm. The two approaches are also complementary in their applicability: the main results 
 of \cite{2016_Bravyi} 
  are most powerful for shallow computations. In particular, their results could only be directly applied for SAT solving if one has access to an exponentially sized quantum computer to begin with, as the computation may be exponentially deep.}. Our results are also applicable to  broader $k$SAT problems, and more generally highlight the characteristics of classical algorithms,
and certain algorithmic methods which can be exploited to provide threshold-free enhancements.  

\textit{3SAT problems:} In SAT problems, we are given a boolean formula $F: \{0,1\}^n \rightarrow \{0,1\}$ over $n$ binary variables $\mathbf{x} = (x_1, \ldots, x_n) \in \{0,1\}^n.$ The task is to find a so-called \emph{satisfying assignment} $\mathbf{x},$ \textit{i.e.}, fulfilling $F(\mathbf{x}) = 1$, if one exists. In 3SAT, $F$ is specified by a set of $L$ clauses $\{ C_j\}$, where each clause specifies three \emph{literals} $\{l^j_1,l^j_2,l^j_3, \}$. Each literal specifies one of the $n$ binary variables $(x)$ or a negated variable $(\overline{x})$, for instance $\{l_1, l_2,l_3\}$ could be $\{ x_3,  \overline{x_5}, x_8 \}$. An assignment of variables $\mathbf{x}$ thus specifies the values of all literals. $C_j$ is satisfied by $\mathbf{x}$ if any of the literals in $C_j$ attains the value 1. A formula $F$, written as $F(\mathbf{x}) = \bigwedge_{j=1}^{L} (l^j_1 \vee l^j_2 \vee l^j_3 )$ using standard logic operator notation, is satisfied by $\mathbf{x}$ if all its clauses are satisfied.

\textit{Classical algorithms:}
Many classical algorithms for 3SAT are known which are significantly faster than brute-force search. Their performance can usually be characterised by a constant $\gamma\in(0,1)$, meaning that the algorithm solves 3SAT in a runtime of $O^\ast(2^{\gamma n})$ \footnote{The $O^{\ast}$ notation suppresses the terms which contribute only polynomially, see, \textit{e.g.}, \cite{2011_Moser}.}. 
One of the best and most famous ones is the algorithm of Sch\"{o}ning \cite{1999_Schoning}. It initializes a random assignment of the variables, then repeatedly finds an unsatisfied clause, randomly selects one literal in that clause, and flips the corresponding variable.
This sampling algorithm terminates once a satisfying assignment is reached, or once this process is iterated $O(n)$ times. Sch\"{o}ning proved that the probability of this algorithm finding a satisfying assignment (if one exists) is at least $(3/4)^n$ which, by iteration, leads to a Monte Carlo algorithm with expected runtime $O^\ast(2^{\gamma_{0} n})$ with $\gamma_{0}:=\log_2(4/3)\approx 0.415$. A significant speedup of the classical algorithm is a reduction of its value of $\gamma$. To study the potential of small quantum computers, we investigate whether small devices suffice
to enhance the algorithm of Sch\"oning by achieving such a reduction. 

\textit{Straightforward hybrid algorithms for small quantum computers:} In \cite{2004_Ambainis} a quantum algorithm inspired by Sch\"{o}ning's method which exploits amplitude amplification \cite{2000_Brassard} was introduced. It solves 3SAT instances with $n$ variables in runtime $O^\ast(2^{{\gamma_{0}} n/2})$ and requires $\approx \beta n$ qubits for a $\beta>2$. 
Given a quantum computer with only $M\approx \beta m$ qubits ($m\ll n$), one has a few options how to create a hybrid algorithm to achieve speedups. A ``bottom up'' approach would be to use the quantum-enhanced algorithm as an $m$-variable instance solver, which is then used as a subroutine in an overarching classical algorithm. For instance, to tackle the problem of $n$ variables, one could sequentially go through all possible partial assignments of $n-m$ variables. Each partial assignment induces a SAT problem with $m$ variables, which could then be solved on the small quantum device. The runtime of such an algorithm is $O^\ast(2^{(n-m) + \gamma_{0} m/2})$ which highlights the \emph{threshold effect}: the hybrid algorithm becomes \emph{slower} than the classical algorithm of Sch\"oning 
if 
$m/n \lesssim 0.74$. 
Roughly speaking, the main problem of hybrid algorithms using a small device as a sub-instance solver is that they break the global structure of the problem exploited by the classical algorithm. This results in  hybrid algorithms whose runtimes interpolate \footnote{By ``interpolate'' we mean that the hybrid algorithm's runtime is characterized by a function $h(n,m)$ of the size of the quantum device (approx. $m$), and the instance size $n$, which is (strictly) monotonically decreasing in $m$, and which roughly attains the classical, and fully quantum runtimes for $m=1$ and $m=n,$ respectively.} between a fully quantum runtime, and something slower than the classical algorithm --  hence the resulting threshold effect. Alternatively, we use a ``top down'' approach, explored in the remainder of the paper, where the subroutines of the classical algorithm, which carry most of the computational work, are identified and quantum-enhanced.

\textit{Our results:} We present a hybrid algorithm that avoids the threshold effect. More precisely, given a quantum computer with $M=cn$ qubits, where $c\in(0,1)$ is an arbitrary constant, our algorithm solves 3SAT with $n$ variables in a runtime of $O^\ast(2^{(\gamma_{0}- f(c)+\varepsilon)n})$, where $f(c)~>~0$ is a constant and $\varepsilon$ can be made arbitrarily small (the details of the function $f(c)$  are given in the supplemental material).
This constitutes a polynomial speedup over Sch\"oning's algorithm 
for arbitrarily small $c$. 

\textit{Algorithm description:} 
Our hybrid algorithm is  
a quantum-assisted version of de-randomized variants of Sch\"oning's algorithm \cite{2002_Dantsin, 2011_Moser}, which we review first. 
 Given a bitstring $\mathbf{x}\in \{0,1\}^n$, let $B_r(\mathbf{x})$  denote the $r-$ball centered at $\mathbf{x},$ \textit{i.e.}, the set of all bitstrings $\mathbf{y}$ differing from $\mathbf{x}$ in at most $r$ positions (\textit{i.e.}, their \emph{Hamming distance} is  $\leq r$). Then, relying on results from coding theory, the space of possible assignments is covered by a number of $r-$balls. Given such a covering set, specified by the centers of the balls, the algorithm sequentially checks whether there exists a satisfying assignment within each of the $r$-balls. This \emph{space splitting} algorithm reduces SAT  to the problem of finding a satisfying assignment within an $r-$ball, called Promise-Ball-SAT (PBS). A deterministic algorithm $\textsc{PromiseBall}(F, r, \mathbf{x})$ for PBS was introduced in \cite{2002_Dantsin}. This is a simple, recursive divide-and-conquer algorithm: on input it takes a formula, specified by a set of clauses (with no more than three literals), a radius, and a center $\mathbf{x}$. The algorithm first checks 
some obvious conditions for (un)satisfiability 
(if $r\leq 0$ and $F(\mathbf{x}) = 0$ or if any clause is empty), or if 
$\mathbf x$ is a satisfying assignment. Otherwise, in the \emph{recursive step,} it finds the first unsatisfied clause $C,$ and calls \textsc{PromiseBall}$(F_{| l =1}, r-1, \mathbf{x})$  for every literal $l \in C$
(we assume that the variables, literals and clauses are enumerated in a pre-specified order). Here, $F_{| l =1}$ denotes the formula  obtained by setting the variable corresponding to $l$ to the value ensuring $l=1$, \textit{i.e.}, all clauses involving $l$ (or $\bar{l}$) are removed (or truncated) \footnote{If a literal appearing in a clause $C = \{l_1,l_2,l_3\}$ evaluates to 1, then $C$ is satisfied regardless of the other variable assignments, and can be removed from the formula. If a literal (say $l_1$) evaluates to 0, then for $C$ to be true we need to satisfy the truncated clause  $C = \{l_2,l_3\}$.}. 

This algorithm solves PBS in time $O^{\ast}(3^{r})$. For comparison, Sch\"{o}ning sampling solves PBS in time $O^{\ast}(2^{r})$.
The overall runtime of the space-splitting algorithm of \cite{2002_Dantsin} can be expressed as a function of the runtime of the PBS-solving subroutine (see supplemental material for details). What is relevant for us, however, is that whenever a PBS solver with runtime $O^{\ast}(2^{r})$ (\emph{e.g.}, randomized Sch\"{o}ning) is used in the space-splitting algorithm, we recover the original Sch\"{o}ning's runtime with $\gamma_{0}\approx 0.415$.

To see how a small quantum device can help, note that  every recursive call in $\textsc{PromiseBall}$ reduces $r$ by $1$. One idea is then to, at the moment $r$ becomes small enough, use a quantum algorithm for PBS instead of a classical recursive call. This leads to a general approach for using small quantum devices to speed up algorithms which recursively call themselves (or other subroutines) with ever 
decreasing instance sizes. 
We will refer to this as \textit{the standard hybrid approach}.

While this is a step in the right direction, there are two types of obstacles to consider. First, since $\textsc{PromiseBall}$ is significantly slower than the basic algorithm of Sch\"{o}ning, this still leads to a threshold problem: $M$ would have to be at least a large fraction of $n$ to gain an advantage. The second issue pertains to the quantum memory requirements of the quantum $\textsc{PromiseBall}$ algorithm, which involves a few technical steps. 
In particular, straightforward quantum implementations of $\textsc{PromiseBall}$ require $\Omega(n)$ ancillary qubits, even if $r$ is small, 
which is too much given the size of our quantum device. 
While it is not difficult to reduce this to $O(r\log n)$, the resulting hybrid algorithm, although avoiding a threshold, would have a very low, in fact sub-polynomial, quantum advantage in the limit of large instances, and not yield a true improvement in terms of $\gamma$.  However, even this problem can be circumvented, using  more involved memory structures combined with specialized algorithms which algorithmically delete unneeded information to reduce the qubit requirements. 
More precisely, we provide a quantum implementation of $\textsc{PromiseBall}$ which requires only $O(r \log(n/r)+r+\log n)$ qubits and which is quadratically faster than classical $\textsc{PromiseBall}$.  As the last step, we use these ideas in combination with a more recent and efficient classical deterministic algorithm for PBS \cite{2011_Moser}, leading to our final algorithm.  

\textit{Space and time efficient quantum algorithm for PBS:}
Note first that any (recursive) algorithm for PBS which `carries' the formulas as explicit inputs has to use sufficiently many bits to represent them. This already violates our constraints on the number of (qu)bits our algorithm should use. Thus, as a first step towards a space-efficient quantum algorithm for PBS, we specify a non-recursive (classical) variant of \textsc{PromiseBall} which does not manipulate the formulas explicitly. Afterwards, we optimize the memory required by the algorithm and turn it into a quantum algorithm using amplitude amplification to obtain the speedup.

The structure of \textsc{PromiseBall} yields a ternary tree of depth $r$, induced by the up to three choices of  literals in the recursive step of the algorithm. Thus, a sequence of choices $s_1,\ldots, s_{r} \in \{1,2,3\}$ specifies a leaf in the tree, and hence the subset of literals whose values have been flipped. In other words, the algorithm \textsc{PromiseBall} induces a mapping from $s_1,\ldots, s_{r}$ to a set of at most $r$ variables to be flipped, denoted $V=\{v_1,\ldots,v_{r'}\}$, where $r'\leq r$. 
The non-recursive algorithm, as a subroutine, computes the list of variables $V$ indexed by the sequence $\vec{s}=s_1,\ldots, s_{r}$, generates the candidate assignments $\mathbf{x}_{V}$ realized by flipping the values of the variables specified by $V$, and checks if they satisfy the formula. This subroutine can be executed in  polynomial time. The overall algorithm simply goes through all $3^r$ sequences  $\vec s$, yielding the runtime $O^{\ast}(3^r)$. 

This can easily be turned into a quantum algorithm $\textsc{QBall}_{12}$, which realizes the mapping:
\EQ{
\ket{s_1,\mydots, s_{r}} 
\ket 0\ket{0}\stackrel{\textsc{QBall}_{12}}{\longrightarrow}
\lefteqn{\underbrace{\phantom{\ket{s_1,\mydots, s_{r}} \ket{V}}}_{\textsc{QBall}_1}}\ket{s_1,\mydots, s_{r}}\overbrace{\ket{V}\ket{F(\mathbf{x}_{V})}}^{\textsc{QBall}_2}, \nonumber \\ \nonumber
}
\noindent
where the first part of the algorithm ($\textsc{QBall}_1$) generates the indices of the variables in $V$, and the second part ($\textsc{QBall}_2$) verifies whether the formula $F$ is satisfied by the assignment $\mathbf{x}_{V}$. The full quantum algorithm for PBS, which we call $\textsc{QBall}$, then uses amplitude amplification to find one sequence $\vec s$ which yields a satisfying assignment, using $O(3^{r/2})$ calls of the $\textsc{QBall}_{1,2}$ subroutines, each with polynomial runtime.

What remains to be seen is how to implement the algorithm space-efficiently. We give the basic ideas here, and full details in the supplemental material. For ease of presentation, we first show how $\textsc{QBall}_2$ can be realized straightforwardly, using many ancillas, and then show how to reduce their number. Such an implementation of $\textsc{QBall}_2$ would simply utilize $n$ additional qubits (one for each variable) and initially assign them the value $\mathbf{x}.$ Then, the circuit would iterate through the registers specifying $V$ (na\"{i}vely requiring $O(r\log n)$ qubits), and introduce (controlled) negations (Pauli $\sigma^x$ gates) to those ancillary qubits selected by the values in $V$. This would finalize the \textit{variable presetting stage}, and set the input to the formula to $\mathbf{x}_{V}$. From there, we would sequentially evaluate each clause, by associating a gate controlled by the variable  qubits corresponding to the variables occurring in the clause. This controlled gate applies the appropriate negations to realize the literals, increasing a counter if a clause is satisfied. After this is iterated for each clause,
 the circuit flips the value of the output qubit only if the counter equals $L$, meaning all the clauses are satisfied. 
 Such a circuit would use $O(\log(L)+n)$ ancillas. 
Since $L =O( \poly(n))$, the key problem is the dependence on $n$. The situation is simplified by noting that although the circuits our algorithms generate depend on $F, \mathbf{x}$ and $r$, we can w.l.o.g. assume that $\mathbf{x} = (0,\ldots, 0)$, and subsume the appropriate negations directly into $F$ \footnote{ 
Let $F_{\mathbf x}$ be the formula obtained by 
negating any literal in $F$ which corresponds to a variable to which $\mathbf x$ assigns the value $1$, so that $F_{\mathbf x}(0, \ldots, 0) = F(\mathbf{x})$. By \emph{subsuming} $\mathbf{x}$ into $F$, we  mean that we use $F_{\mathbf x}$ instead of $F$}. Next, observe that since the clauses are evaluated sequentially, we actually only require three variable-specific ancillas, specifying the values of the variables appearing in the current clause:  for each clause, each of the three ancillas correspond to the three variables occurring in the given clause. 
The variable presetting stage is now done individually for each clause $C_j$: before clause evaluation, the circuit iterates through the $V-$specifying register, and flips the $k^{th}$ ancilla if the specification matches the  $k^{th}$ variable within the clause $C_j$. The three ancillas are uncomputed after $C_j$ is evaluated, and can be reused. This requires at most $O(\log n)$ additional qubits. 

The algorithm $\textsc{QBall}_1$ is more involved, but again one of the main savings utilizes the observation that few ancillas, which can be reused, suffice to determine whether a given clause is satisfied. The rough description of our algorithm is as follows. $\textsc{QBall}_1$ comprises a main loop which sequentially adds one variable specification $v_i$ to the already specified set as follows: the $i$\textsuperscript{th} circuit block takes the specifications of the first $i-1$ variables $v_1, \ldots, v_{i-1}$ as inputs, iterates through all the clauses, and evaluates each clause $C_j$ (in a manner similar to  $\textsc{QBall}_2$), using the values $v_1, \ldots, v_{i-1}$ to correctly preset the clause-specific input. If the clause is not satisfied, it uses the value of $s_i$ to select the specification of one variable occurring in $C_j$, taking into account that variables which have already been flipped cannot be selected again, and storing this specification as $v_i$. 
To ensure reversibility, additional counters have to be used, but these can be uncomputed and recycled (see supplemental material for details).
The final compression relies on a more efficient encoding of $V,$ which, as an ordered list  would use $O(r\log n)$ qubits. Since the ordering does not matter, instead of storing the positional values, we can store the relative shifts $v_1, v_2-v_1,\ldots, v_{r'}-v_{r'-1}>0$ of a \emph{sorted version} of the list,  
using no more digits than necessary and a separation symbol to denote the next number.  
Since these values add up to at most $n$, it is see that  $O(r\log (n/r) + r)$ qubits suffice for this encoding. 
This structure indeed encodes a set, since the initial ordering is erased. 
However, straightforward algorithms that use such encodings of sets instead of ordered lists encounter problems with reversibility. To illustrate this, note that in the process of adding a new variable to $V$, one must realize the two steps $\ket{\{v_1, \ldots, v_{i-1} \} }  \ket{0} \mapsto  \ket{\{v_1, \ldots, v_{i-1} \} }  \ket{v_i} $, \textit{i.e.}, finding the new element, and 
 $\ket{\{v_1, \ldots, v_{i-1} \} }  \ket{v_i} \mapsto  \ket{\{v_1, \ldots, v_{i} \} }  \ket{0} $, 
 \textit{i.e.}, placing it into the set, and, critically, freeing the ancillary qubit for the next step. However, this is irreversible, since the information which element was added last is lost. The full ordering information requires additional $O(r \log (r))$ additional qubits, which would nullify all advantage. 
 Of course, one could instead realise the reversible operation $\ket{\{v_1, \ldots, v_{i-1} \} }  \ket{v_i}\ket 0 \mapsto \ket{\{v_1, \ldots, v_{i-1} \} } \ket{v_i}  \ket{\{v_1, \ldots, v_{i} \} }  $, followed by applying the inverse of the entire circuit up to this point to uncompute $\ket{\{v_1, \ldots, v_{i-1} \} }  \ket{v_i}$, but this would result in an exponential instead of polynomial runtime of $\textsc{QBall}_1$.  
 We circumvent this issue by splitting $V$ into $O(\log r)$ sets of sizes $1,2,4,\ldots$, and the loading of each larger block is followed by an algorithmic deletion of all smaller blocks. 
This ensures that the overall number of qubits needed for this encoding is still just $O(r\log(n/r) + r),$ at the cost of $2^{O(\log(r))}$ additional steps, which is still only polynomial. These structures and primitives lead to an overall space- and time-efficient implementation of $\textsc{QBall}_1$ (see the supplemental material for details).
Combining these subroutines with quantum search over the vector $\vec{s}$, we obtain the algorithm \textsc{QBall}, solving PBS in time $O^\ast(3^{r/2})$ and using $O(r\log (n/r) + r + \log n )$ qubits.

\textit{Hybrid algorithm for 3SAT:} The runtime of \textsc{QBall} not only quadratically beats that of  \textsc{PromiseBall}, but also outperforms the randomized algorithm of Sch\"{o}ning ($2^{r}$ vs. $2^{ \log_2(3) r /2  } \approx 2^{0.79 r}$). 
Since our quantum devices are size-limited, they can only solve PBS for sufficiently small $r$. In principle, we could use the standard hybrid approach for $\textsc{PromiseBall}$, \emph{i.e}, call \textsc{QBall} instead of \textsc{PromiseBall} deep in the recursion tree when $r$ is sufficiently small. But, as mentioned earlier, this still leads to a threshold problem. 
This is resolved by considering an improved classical deterministic algorithm for PBS \cite{2011_Moser}, where coding theory is applied once more to cover the space of the choice vectors $\vec{s}$ with covering balls. This yields an algorithm with runtime $O^\ast((2+\epsilon)^r)$, where $\epsilon$ can be chosen arbitrarily small -- in other words, the runtime of this algorithm for PBS essentially matches the runtime of Sch\"{o}ning. While the details of this algorithm are not important here, the critical point is that, like the original \textsc{PromiseBall} algorithm, it recursively calls itself to solve PBS with ever smaller values of $r$ (each time reduced by a quantity which depends on $\epsilon$). Thus, the standard hybrid approach can be applied. Since $\textsc{QBall}$ beats the runtime of %Sch\"{o}ning's algorithm, 
this improved classical algorithm for PBS, 
the hybrid algorithm is faster than Sch\"oning's algorithm, and, unlike using \textsc{PromiseBall}, there is no threshold induced by sub-optimal classical routines. 

To estimate the runtime of our hybrid algorithm, note that since the quantum algorithm only requires $O(r\log (n/r) + r + \log n)$ qubits, a device with $M=cn$ qubits can solve PBS for $r=f(c)n$ for a function $f$ (see supplemental material for details). Since the hybrid algorithm replaces a classical subroutine of runtime $O^\ast(2^r)$ with a quantum subroutine of runtime $O^\ast(3^{r/2})$ in a recursion tree below depth $r=f(c)n$, the runtime of the hybrid algorithm beats that of the classical one by a factor of $O^\ast( (\sqrt3/2)^{f(c)n})$. Thus, the combined runtime of the hybrid algorithm is $O^\ast(2^{(\gamma_{0}+\varepsilon - 0.21 f(c))n} )$, as claimed.

\textit{Conclusions and Outlook:}
We have shown that a small quantum computer can speed up the solving of relevant computational problems of significantly larger size. While obvious for structureless problems (\emph{e.g.}, unstructured search), when considering algorithms which use the problem's structure, such as in the case of Sch\"oning's algorithm, speedups are non-trivial: the way the problem is partitioned must maintain the structure which is exploited by the classical algorithm to avoid thresholds. 

In terms of the broader underlying question, our work complements \cite{2016_Bravyi}, where classical-assisted quantum algorithms are considered. Those techniques are applicable in the opposite regime when the quantum computer is almost the size of the problem, and the computation is of comparatively small depth.

Our algorithm achieves a significant speedup, as given by a reduction of the rate $\gamma$, which is the relevant performance parameter,  
characterizing runtimes of the form $O^\ast(2^{\gamma n})$.
The speedup we provide here holds relative to the algorithm of Sch\"oning. Our results however generalize to certain other algorithms which are based on Sch\"oning's algorithm (\emph{e.g.}, \cite{2002_Hofmeister,2010_Iwama}
), since those are mainly achieved by using better initial assignments, and to the variants tackling $k$SAT ($k>3$). 
Historically, the best solvers with provable bounds are based either on the ideas of Sch\"{o}ning, or alternatively, the approach of \cite{2005_Paturi}, which includes, to our knowledge, the current classical record holder \cite{2014_Hertli}. 
It would be interesting to see whether this second class of SAT algorithms is also amenable to the types of enhancements we achieve here.

Apart from fundamental interest, the question of this work is becoming increasingly more relevant given the current progress in prototypes of small quantum computers \cite{2017_IBM,2018_Google,2018_Intel}. 
We assume an error-free ideal scenario, while in practice this may be one of the bottlenecks to exploit small devices \cite{2018_Preskill}. Thus, it would be particularly interesting to see how the presented algorithm behaves given noise, and develop methods to decrease the number of gates, thereby increasing tolerance.

\textit{Acknowledgements.--}
VD acknowledges the support from the Alexander von Humboldt Foundation. JIC acknowledges the ERC Advanced Grant
QENOCOBA under the EU Horizon 2020 program (grant
agreement 742102).

%merlin.mbs apsrev4-1.bst 2010-07-25 4.21a (PWD, AO, DPC) hacked
%Control: key (0)
%Control: author (72) initials jnrlst
%Control: editor formatted (1) identically to author
%Control: production of article title (-1) disabled
%Control: page (0) single
%Control: year (1) truncated
%Control: production of eprint (0) enabled
%

%\end{document}

\onecolumngrid

\newpage
\pagebreak
\newpage
\pagebreak

\section*{Supplemental Material}
\label{SM}

In the supplemental materials we provide the details referred to in the main text. In Section~\ref{sm:qPBS}, we provide a high level description of the algorithm, along with the key ideas. The important aspects of the algorithm are rather straightforward, however, for completeness we provide a fully detailed exposition of the central quantum algorithm \textsc{QBall}. 
In Section~\ref{sm:runtimehybrid}, we provide the details of the full hybrid algorithm and a detailed runtime analysis.

\subsection{Detailed quantum algorithm for PBS}
\label{sm:qPBS}
In this section, we describe the quantum algorithm for PBS. Let $r\geq 0$ be an integer and $F$ a 3-SAT formula over $n$ variables. The problem is to decide whether there exists a satisfying assignment of $F$ with hamming weight at most $r$.

As explained in the main text, the most involved part of the algorithm is a subroutine called $\textsc{QBall}_1$ which takes as input  $s_1,\ldots,s_r\in\{1,2,3\}$ and outputs a set $V=V(\vec s)$ of at most $r$ variables with the property that  $F$ has a satisfying assignment with hamming distance at most $r$ if and only if there exists an $\vec s\in\{1,2,3\}^r$ such that the assignment $x(V(\vec s))$, obtained by setting all variables in $V(\vec s)$ to $1$ and all other ones to $0$, is a satisfying assignment.

This algorithm is a quantum version of the non-recursive variant of the \textsc{PromiseBall} algorithm described in the main text, which is quantized using amplitude amplification. We will therefore first develop a classical \emph{reversible} circuit (which depends on $F$ and $r$), using at most $O(r\log(n/r) + r + \log n)$ ancilla bits that maps $\vec s$ to a space-efficient encoding of the set $V(\vec s)$. Although the algorithm below is, for now, purely classical (and reversible), we will nevertheless use bra/ket notation for simplicity.

\subsubsection{Key ideas of the quantum algorithm for the Promise Ball problem}
Before going into the particulars of the algorithm, it is worth-while to highlight the key ideas we utilize, and explain why we go into such detail to explain a relatively simple algorithm.
As mentioned, given the radius $r$ and ball center $\mathbf{x}$ (this can be shifted to all-zero by modifying the formula), our algorithm can be summarized as follows:

\begin{figure}[H]
\begin{algorithmic}[1]
%\Procedure{$A$}{$F$}%\Comment{Is this satisfied}
\State Set $V = \emptyset$.
\State For $i=1$ to $r$ do begin
\State \indent Set $\mathbf{x}_V$ to be the assignment where the variables in $V$ gave been flipped.
\State \indent Find first clause in $F$ which is not satisfied by $\mathbf{x}$' 
\State \indent If such a clause exists, use $s_i$ to select the next variable $v_i$. Else set $v_i$ to  a dummy variable (index beyond $n$). 
 \State \indent Update $V\mathop{:} = V \cup \{v_i\}$.
\State end.
\end{algorithmic}
\end{figure}

The output $V$ of the above process collects specifications of $r$ variables, and the variables with indices below or equal to $n$ specify the satisfying assignment, if one exists. We use dummy variables as this avoids the need for controlled operations which would change the algorithm behaviour (\emph{e.g.} terminate the loop once no suitable clause is found), which end up requiring more memory to be implemented reversibly.
The majority of subtleties in our algorithms pertain precisely to this: utilizing as few systems as possible, while maintaining reversibility.

As mentioned in the main text, one of the key issues is dealing with the size of the representation of $V$. If $V$ is represented as an ordered list, we end up using $r \log(n)$ qubits, which is too much for our purposes.
On the other hand, $V$ can be implemented as a set, forgetting the order, in which case we require only $r\log(n/r)$ qubits (we only need to store the shifts between the indices of the variables once they are ordered). This is sufficient for our purposes, however, now the process of adding one variable to the set, which occurs at each step $1\ldots,r$, is no longer reversible -- more specifically, a reversible operation produces registers which still, implicitly, contain information about the positions. 
Since all the processes we utilize will be used in an overarching amplitude amplification process, we must use only reversible operations -- no measurements with classical feedback can be used.
If these registers containing the implicit position information are not re-used, we end up using too much memory. On the other hand, to re-use those registers, since measurement is not an option, we need to un-compute information, which may be computationally expensive.

To exemplify the problem, it is possible to implement the operation
\EQ{
\ket{Set\ V_{i-1}} \ket{v_i} \ket{0} \mapsto \ket{Set\ V_{i-1}} \ket{v_i} \ket{Set\ V_{i-1} \cup \{ v_i\} }, }
followed by the deletion of the first register in the state $ \ket{Set V_{i-1}}$, by reversing whatever circuit was used to generate it. 
This would enable us to realize the transform 
\EQ{
\ket{Set\ V_{i-1}} \ket{v_i} \ket{Set\ V_{i-1} \cup \{ v_i\} }  \mapsto \ket{Set\ V_{i}} \ket{0} \ket{0}
}
where $V_i = V_{i-1} \cup \{ v_i\}$.  This constitutes an algorithmic deletion procedure, where we have deleted all unnecessary information.
However, it is easy to see that the process calling the inverse of the entire circuit up to step $i-1$, which is (recursively) done for each $i$, leads to a run-time which will be exponential in $i$. Since $i$ ranges to $r$, we end up with an algorithm with an exponential number of steps in $r$. 
In the interesting cases when $r$ is a fraction of $n$, this becomes (significantly) more expensive than the overall classical algorithm.

In the process detailed below, we circumvent this problem by using a more structured memory which allows for more efficient deletion. 
Specifically, we will split the memory into approximately $\log(r)$ sets, each storing twice as many variable indices, so of sizes $ l = 1,2, 4, \ldots$, in total storing all $r$ variable specifications. This memory structure is still space-efficient, in that it still requires $O(r \log(n/r))$ (qu)bits.
However, such a structure can be used for efficient deletion.
The basic idea can be illustrated by the following algorithm, which is less efficient than our final proposal given later, but easier to analyze.

For each size $l$, we will utilize two memory blocks. The moment the two memory blocks of some size $l'$ are filled with variable indices, they are merged and stored into the larger memory block of size $2 l'$.
These two memory blocks of size $l'$ are now algorithmically deleted by inverting the processes which filled them. 
Again, like in the first algorithmic deletion procedure, the algorithm is specified using recursive calls.  However, in the former case, the depth of the recursion was  $r$. In the more efficient algorithm, the recursion depth is up to the number of differing block sizes. Since we use $O(\log(r))$ different-sized blocks, the overall computational cost is proportional to $O(2^{\log(r)}),$ so  polynomial, and not exponential, in $r$. 

The more efficient algorithms below are based on the same idea, and presented in more detail.
The high-level descriptions of all the subroutines required to realize all the elementary steps to execute the overall algorithm are provided in Table \ref{tab}.
Each sub-routine must satisfy the same two critical properties we just discussed for the highest level of the algorithm description. First, each subroutine must be economic with respect to how many ancillary (qu)bits are used. Second, since each subroutine is used many times, it is also critical that the ancillary qubits are always reset (uncomputed), so they can be reused, and that the deletion is sufficiently efficient. 
To make sure both properties are maintained, we provide all the subroutine description to full detail.

\begin{table}

\begin{tabular}{|l | p{13cm}|}
\hline
Main subroutines: & \\
\hline
$\calculate_i$ & $\ket{s_i}\ket{\EffEnc V_{i-1}} \ket{0} \mapsto \ket{s}\ket{\EffEnc V} \ket{v_{i}}$

\vspace{3pt}\\
& Given a trit $s_i$ and (an efficient encoding of) the current set  of variables,  produces the next variable $v_{i}$  to be flipped.\\
\hline
$\merge_i$ & $\ket{\vec s}\ket{\EffEnc V_{i-1}} \ket{v_{i}} \mapsto \ket{\vec s}\ket{\EffEnc V_i}$
\vspace{3pt}\\
& Adds the next variable to the set while maintaining efficient encoding, and uncomputes the value of the added variable.\\
\hline
Key ancillary subroutines: & \\
\hline
$\extract_k$ & $\ket{\enc S}\ket{j}\ket0 \mapsto \ket{ \enc S}\ket j \ket{y_j}$
\vspace{3pt}\\
& Extracts the $j^{th}$ element (w.r.t. ascending ordering) from the encoding of the set $S$ of known size $k$.  Critical subroutine in $\calculate_i$.\\
\hline
 $\shift_k$ & $\ket{\enc S}\ket j\ket 0 \mapsto \ket{\enc S}\ket j\ket{y_j-y_{j-1}}$\vspace{3pt}\\
 & Used in $\extract_k$. The set-encoding stores shifts between neighbouring \\ &elements, and 
 this subroutine extracts them. In essence, it counts special separation characters delineating different numbers in the set encoding.\\
 \hline
 $ \contains_k$ & $\ket{\enc S}\ket v\ket 0 \mapsto \ket{\enc S}\ket v \ket{v\in S?}$\vspace{3pt}\\
 &Checks if $v$ is present in the set $S$, given encoding of $S$. Required for checking if   a clause is satisfied when flipping the values of literals  if they occur in $S$,  and also when selecting the next variable, as we can only select those which have not been chosen  previously. Uses $\extract_k$.\\
 \hline
$\ch_{j,i}$ & $\ket{\EffEnc V_{i-1} }\ket 0 \mapsto \ket{\EffEnc V_{i-1}}\ket{C_j(V_{i-1})}$\vspace{3pt}\\
 &Checks if the $j^{th}$ clause is satisfied when the variables in the set $ V_{i-1}$ are flipped.  One of the key subroutines.\\
 \hline
  $\select_i$ & $\ket{\EffEnc V_{i-1}}\ket {s_i} \ket{L+1-j}\ket 0 \mapsto \ket{\EffEnc V_{i-1}}\ket{s_i} \ket{L+1- j}\ket{v_{j,i}}$\vspace{3pt}\\
  &  Once a clause $C_j$ is identified to  not  be satisfied, this subroutine uses the choice $s_i$ and  
  the set of already flipped variables to select a variable  in $C_j$  \\ & will be flipped next. \\
  \hline
  $\append_k$ &  $\ket v\ket{\enc S}\ket{0} \mapsto \ket v\ket{\enc S\cup \{v\}}$
   \vspace{3pt}\\
    &  Adds an element to the set $S$ and computes the encoding of the new set. The added element is assumed to be larger than the largest element in $S$. \\
\hline
  $\union_{k_1,k_2}$ & $\ket{\enc S_1}\ket{S_2}\ket{0} \mapsto \ket{\enc S_1}\ket{\enc S_2}\ket{\enc S_1\cup S_2}$
   \vspace{3pt}\\ &  Generates an efficient encoding of the union of two efficiently encoded sets. In essence, uses $\extract$ on both sets to compare elements,  storing the smaller, and keeping a  position counter for both. Uses $\append_k$.\\
    \hline
\end{tabular}
\caption{\label{tab}Summary of subroutines for the space-efficient reversible implementation to solve PBS. The two main subroutines ran in succession find the next variable, add it to the set, and, critically, free the ancillary registers so they can be reused.}\label{tab:subroutines}
\end{table}

\subsubsection{Algorithm overview}

The basic idea for the algorithm to calculate $V(\vec s)$ from $\vec s$  consists of a main loop going through $i=1,\ldots,r$, and attempts to find another variable $v_i$ to be added to $V$ in the following way. The algorithm finds the first clause $C_j$ that is not satisfied by setting all variables in $V$ to $1$ (we assume that there is some predefined order of the clauses $C_1,\ldots,C_L$). We then select the next variable $v_i$ to be the $s_i$\textsuperscript{th} variable in $C_j$ that is not already in $V$. If no such variable exist (\textit{i.e.}, the entire clause $C_j$ is already in $V$) or there exist no unsatisfied clause under the current $V$, then the algorithm adds a dummy variable $x_{n+i}$ to $V$. 
Na\"{i}vely and without any further processing, this would lead to an encoding of $V$ that simply stores the ordered list $v_1,\ldots,v_r$, taking $O(r\log n)$ bits in total. As mentioned in the main text, this would ultimately lead to too much of a qubit requirement of the corresponding quantum algorithm to give a strong speed up. We therefore first describe how to efficiently encode a set of variables.

 Let $S=\{y_1,\ldots,y_k\}$ be a set, where $0<y_1<\cdots < y_k \leq 2n$ are integers. Define $\ket{\enc S}$ to be a sequence of at most $O(k\log(n/k) + k)$ trits encoding $(y_1, y_2-y_1,\ldots, y_k-y_{k-1})$, where each number is encoded on the binary subspace of the trits using no more digits than necessary and the third symbol of the trit subspace serves as a separation character between successive numbers. Note that since the numbers add up to at most $n$, the number of trits required for this encoding is indeed at most
\begin{align}
	\lceil \log_2 y_1\rceil + \lceil \log_2 (y_2-y_1)\rceil+ \cdots+\lceil \log_2 (y_k-y_{k-1})\rceil + 2k &\leq \log_2 y_1+\log_2 (y_2-y_1) + \cdots + \log_2 (y_k-y_{k-1})  + 3k \\
	&\leq k\log_2\frac{y_k}{k} +3k\\
	&\leq k\log_2\frac nk + 4k,
\end{align}
where the second line follows from concavity of the logarithm. Note that this is significantly less than the na\"{i}ve encoding of $S$ which uses $O(k\log n)$ bits. For concreteness, we will assume that after the last separation character, the remaining trits are in $0$.

This gives us an efficient way of storing a set of at most $r$ elements. However, encoding $V$ as $\ket{\enc V}$ at all stages of the algorithm is problematic, because reversibility would be lost (indeed, the map $\ket{\enc V_{i-1}}\mapsto \ket{\enc V_i}$ is not injective). We therefore store the set of variables during the algorithm as follows. Let $V_i=\{v_1,\ldots,v_i\}$ be the set of variables obtained after the $i$\textsuperscript{th} round. Note that $|V_i|=i$. We now divide $V_i$ into subsets $V_{i,1},\ldots, V_{i,m_i}$ of sizes which are powers of two corresponding to the binary expansion of $i$. For example, if $i=13$ (binary expansion $1101$), we have $V_{i,1} = \{v_1,v_2,v_3,v_4,v_5,v_6,v_7,v_8\}$, $V_{i,2} = \{v_9,v_{10},v_{11},v_{12}\}$ and $V_{i,3} = \{v_{13}\}$. 

To efficiently store $V_i$, we efficiently encode the set as $\ket{\EffEnc V_i}:=\ket{\enc V_{i,1}}\ldots\ket{\enc V_{i,m_i}}$. Note that the number of bits required for $\ket{\EffEnc V_i}$ is at most
\begin{align}
	\sum_{l=0}^{\lceil \log_2 r\rceil } \left( 2^l\log_2\frac n{2^l} + 2^{l+2}\right) &= 
	\sum_{l=0}^{\lceil \log_2 r\rceil } 2^l\log_2\frac nr
	+ 
	\sum_{l=0}^{\lceil \log_2 r\rceil } 2^{l+2}
	+
	\sum_{l=0}^{\lceil \log_2 r\rceil } 2^l \log_2\frac r{2^l}
	\\
	&\leq 4r\log_2\frac nr + 16r + 
	2r\sum_{l=0}^{\lceil \log_2 r\rceil} \frac{1}{2^{\lceil \log_2 r\rceil -l}} (\lceil \log_2 r \rceil -l)
	\label{eq:numberQubitsInt}\\
	&\leq 4r\log_2\frac nr + 20r \label{eq:numberQubitsInt2} \\
	&= O(r\log(n/r) + r), 
\end{align}
where in \eqref{eq:numberQubitsInt}, the first two terms follow from $
	\sum_{l=0}^{\lceil \log_2 r\rceil } 2^l \leq 4r
$ and the last term follows from $2^{\lceil \log_2 r\rceil} \geq r > 2^{\lceil \log_2 r\rceil-1}$, and \eqref{eq:numberQubitsInt2} follows from
\begin{equation}
	\sum_{l=0}^{\lceil \log_2 r\rceil -l } \frac{1}{2^{\lceil \log_2 r\rceil}} (\lceil \log_2 r \rceil -l) < \sum_{j=0}^\infty \frac j{2^j} = 2.
\end{equation}

For a given $i\in\{1,\ldots, r\}$, the circuit is divided into a calculate phase and a merge phase. The calculate phase performs 
\begin{equation}
	\calculate_i\ket{s_i}\ket{\EffEnc V_{i-1}}\ket 0 
	=\ket{s_i} \ket{\EffEnc V_{i-1}}\ket{v_i},
\end{equation}
\textit{i.e.}, finds the next variable to be added to $V$. 
The merge phase performs
\begin{equation}
	\merge_i\ket{\vec s}\ket{\EffEnc V_{i-1}}\ket{v_i}
	=
	\ket{\vec s}\ket{\EffEnc V_i},
\end{equation}
\textit{i.e.}, converts the efficient encoding of the old set into the efficient encoding of the new set. It achieves ths by using a series of set merge operations comprising calculating the representation of the set-union and the uncomputation of original the set representations.

While the main ideas of both phases are straightforward, the primary challenge is to develop these algorithms in the form of circuits that are both reversible and at the same time require few ancillas. For the latter, it is imperative that at the end of any subroutine, all ancillas, which can be assumed to be initially in the $0$ state, are returned to $0$ so that they can be reused. Note that not achieving this would increase the ancilla requirement for each call of any such subroutine and would likely result in the algorithm requiring too many ancillas. For convenience, we therefore introduce the following notion.

\begin{definition}
	Let $f:\{0,1\}^q \rightarrow \{0,1\}^q$ be a bijection. We say that $f$ can be \emph{implemented reversibly using $l$ ancillas and $g$ gates} if there exists a (classical) reversible  circuit of at most $g$ elementary gates which implements the map
	\begin{equation}
		\ket{x}\ket{0}^{\otimes l} \mapsto \ket{f(x)}\ket{0}^{\otimes l}.
	\end{equation}
\end{definition}

\begin{proposition}\label{prop:phases}\
	\begin{enumerate}[(i)]
	\item \label{prop:phasespart1} $\calculate_i$ can be implemented reversibly using $O(\log n)$ ancillas and $O(\poly(n))$ gates.
	\item \label{prop:phasespart2} $\merge_i$ can be implemented reversibly using $O(r\log(n/r)+r+\log n)$ ancillas and $O(\poly(n))$ gates.
	\end{enumerate}
\end{proposition}

In the next two subsections, we prove Proposition~\ref{prop:phases} by describing each phase in detail. Table~\ref{tab:subroutines} gives an overview of the key subroutines involved. 

\subsubsection{The calculate phase}

In this part of the algorithm for a given $i\in\{1,\ldots,r\}$, we calculate the $s_i$\textsuperscript{th} variable in the first unsatisfied clause.
This part of the algorithm has  two components. First, finding the first clause that is not satisfied, and second, selecting the corresponding variable.

On input, we have $\ket{\EffEnc V_{i-1}}\ket{s_i}$. Note that the sizes of the sets are known in advance. We first show how given a set, we can efficiently and reversibly extract its elements using few bits. 	For a set $S=\{y_1,\ldots,y_k\}$ of integers $0<y_1<\cdots <y_k\leq 2n$ of known size $k$, define the reversible operation $\extract_k\ket{\enc S}\ket{j}\ket0 = \ket{ \enc S}\ket j \ket{y_j}$ for $j\in\{1,\ldots,k\}$. We will assume that any set $S$ we consider has at most $ r\leq n$ elements.

\begin{lemma}\label{lem:extract}
	$\extract_k$ can be implemented reversibly using $O(\log n)$ ancilla bits and $O(\poly(n))$ gates. 
\end{lemma}
\begin{proof}
	We first show how to implement the operation $\shift_k\ket{\enc S}\ket j\ket 0 = \ket{\enc S}\ket j\ket{y_j-y_{j-1}}$, where by convention $y_0=0$. This operation comprises three parts: first compute the position of the $(j-1)$\textsuperscript{st} separation character in $\ket{\enc S}$ and the number of trits to the $j$\textsuperscript{th} separation character (we call this operation $U_{1,k}$), second copy the relevant trits to a separate register (we call this operation $U_2$), and third uncompute all other ancillas by running $U_{1,k}^{-1}$. 

We introduce three counters, called the block position counter, the length counter and the $j$-counter, going from $0$ to $\lceil k\log_2 (n/k)\rceil + 4k$, $0$ to $\lceil \log_2 n\rceil+1$, and $0 $ to $k$ respectively, and all initially set to $0$. Note that all three counters use no more than $O(\log n)$ bits. To implement $U_{1,k}$, we do the following steps.

First, we add $k$ to the $j$-counter to set it to $k$. Then, for each trit in $\ket{\enc S}$ sequentially, we do the following: Controlled on the $j$-counter being larger than $k-j+1$, increase the block position counter by one; controlled the trit being a termination symbol, decrease $j$-counter by one.
This ends with the $j$-counter being $0$, and the block position counter containing the position of the $j$\textsuperscript{th} block. 

Note, the implementation of operations controlled on the state of counters can be done by a number comparison,
and the latter can be implemented reversibly with one ancilla bit. 
 
Next, 
for each trit in $\ket{\enc S}$, we sequentially do the following: If the $j$-counter is equal to $j-1$, increase the lenght counter by one; if the trit is a termination symbol, increase the $j$-counter register by one. 

Since the length counter is only increased in the $j$\textsuperscript{th} block, it stores its length. The $j$-counter now has the (known) value $k$ and we can reset it to $0$.

This implements $U_{1,k}$, 
which performs
\begin{equation}
	U_{1,k}\ket{\enc S}\ket j \ket 0 \ket 0 =
	\ket{\enc S}\ket j\ket{\pos(S,j)}\ket{\len(S,j)}
\end{equation}
and resets all ancillas to $0$, where $\pos(S,j)$ is the position of the $(j-1)$\textsuperscript{th} separation character and $\len(S,j)=\pos(S,j+1)-\pos(S,j)$ is the number of trits to the $j$\textsuperscript{th} separation character.  

To implement $U_{2}$, it will be convenient to define the family of $U(j,l)$ copy-unitaries. 
Each such unitary is a sequence of $l$ CNOT gates, which copy the trits from position $j+1$ to position $j+l$ as the last $l$ bits of the shift-output register. Each such unitary costs only $l$  CNOTS. 
The  $c-U(j,l)$ is the controlled family of such unitaries, which is controlled on the states of the block position and the length counter, 
and the corresponding $U(j,l)$ is only activated if the block position counter and the length counter are equal to $j$ and $l$, respectively.
For completeness, these are only well-defined if $j + l \leq \lceil k \times log_{2}(n/k) \rceil + 4k$, \textit{i.e.}, the lenght of the register containing $\ket{\enc S}$. 
If the labels are out-of-bounds, we substitute them with identities. 
Although they copy trits, they will only be used as to act on the binary subspace, as we will not be copying the termination symbols. 
$U_{2}$ is then given by 
\begin{equation}
	U_{2} = \prod_{j = 1}^{\lceil log_{2}(n/k) \rceil + 4k } \prod_{ l = 1}^{\lceil log_2 n \rceil+1} c-U(j,l).
\end{equation}
Note that only one $U(j,l)$ will be activated, namely the one specified by the position and length registers.
In total, $U_{2}$ contains $O((\log(n/k) + k) \log(n) ) = O(\log(n)^2)$ gates, and uses no ancillas.
Applying $U_{1,k}^{-1}$ then completes the implementation of $\shift_k$. 

Finally,  $\extract_k$ can be implemented by introducing an additional counter from $1$ to $k$, and $k$ calls to controlled-$\shift_k$, each controlled on the counter being $\leq j$, followed by incrementing that counter. This way, each call adds $y_l-y_{l-1}$ to the output register if $l\leq j$. This finally leaves the output register in $\ket{y_j}$ and all ancillas in $0$, as desired.  
\end{proof}

This extraction subroutine can be used to check if a given variable index $v\in\{1,\ldots,2n\}$ is contained in a set of known size. To this end, define the (reversible) operation $\contains_k$ which performs $\contains_k\ket{\enc S}\ket v\ket 0 = \ket{\enc S}\ket v \ket{v\in S?}$, where the last bit is $1$ if $v\in S$ and $0$ otherwise. 

\begin{lemma}\label{lem:contains}
	$\contains_k$ can be implemented reversibly using $O(\log n)$ ancilla bits and $O(\poly(n))$ gates. 
\end{lemma}
\begin{proof}
 We introduce a counter from $0$ to $k$ and $O(\log n)$ additional ancilla bits. Then, for each $j=1,\ldots,k$, we use Lemma~\ref{lem:extract} to extract the $j$\textsuperscript{th} element, controlled on the counter being in $0$. We then check whether the extracted element is equal to $v$, and store the result in a separate temporary result bit. We then uncompute the controlled extraction and, controlled on the result bit being in $1$, increment the counter. After doing this for all $j=1,\ldots,k$, the result bit is in $1$ if the element is in $S$ and $0$ if not. We apply a CNOT of the temporary result bit to the output bit and then run the inverse of the entire circuit up to that point to reset the ancillas. 
\end{proof}

Next, for a clause $C_j$, let
\begin{equation}
	C_j(V) =  \begin{cases}
		1 & C_j \text{ satisfied by $x(V)$} \\
		0 &  \text{otherwise}.
	\end{cases}
\end{equation} 
We furthermore define a subroutine $\ch_{j,i}$ 
which performs 
\begin{equation}
\ch_{j,i}\ket{\EffEnc V_{i-1} }\ket 0 = \ket{\EffEnc V_{i-1}}\ket{C_j(V_{i-1})}.
\end{equation}
 Note the dependence of $\ch_{j,i}$ on $i$ reflects that the routine is adjusted slightly for different set sizes appearing in $\ket{\EffEnc V_i}$ (which are known in advance). 

\begin{lemma}\label{lem:check}
	$\ch_{j,i}$ can be implemented reversibly using $O(\log n)$ ancilla bits and $O(\poly(n))$ gates.
\end{lemma}
\begin{proof}
		We introduce three additional ancilla bits storing the values of the three variables in the clause. Then, for each variable and each set in $\ket{\EffEnc V_{i-1}}$, we use Lemma~\ref{lem:contains} to check if the variable is in the set. The result is copied onto the variable bit using a CNOT, and the $\contains_k$ operation is then reversed. Note that the sets can be assumed to be disjoined. Then, the variable bits are flipped according to the literals in the clause. The clause is then evaluated and the result stored in a separate bit. We then apply the inverse of the circuit to reset the ancillas. 
\end{proof}

Next, we show how to select a variable from a given clause. For a clause $C_j$ with variables $x_a,x_b,x_c$ with $0<a<b<c\leq n$, let $v_{j,i}(V_{i-1},s_i)$ be the  $s_i$\textsuperscript{th} smallest number in $\{a,b,c\}\backslash V_{i-1}$. If $s_i>|\{a,b,c\} \backslash V_{i-1}|$, then $v_{j,i}(V_{i-1},s_i)= n+i$. We now define the reversible operation $\select_i$ which performs 
\begin{equation}
\select_i\ket{\EffEnc V_{i-1}}\ket {s_i} \ket{L+1-j}\ket 0 = \ket{\EffEnc V_{i-1}}\ket{s_i} \ket{L+1- j}\ket{v_{j,i}(V_{i-1},s_i)} 
\end{equation}
 for $j=1,\ldots,L $ and 
\begin{equation} \select_i\ket{\EffEnc V_{i-1}}\ket {s_i} \ket 0\ket 0 = \ket{\EffEnc V_{i-1}}\ket{s_i}\ket 0 \ket{n+i}.
\end{equation} 
\begin{lemma}
	$\select_i$ can be  implemented reversibly using $O(\log n)$ ancilla bits and $O(\poly(n))$ gates. 
\end{lemma}
\begin{proof}
	First, we add $n+i$ to the result register controlled on the counter input register being in $0$. Then, for each $j$, we do the following operations controlled on the counter input register being in $L+1-j$: Use Lemma~\ref{lem:contains} to check which variables are contained in $V_{i-1}$ and store the result in three variable ancilla bits (similarly to the implementation $\ch_{j,i}$ in Lemma~\ref{lem:check}). Then, for each combination of values of the variable ancilla bits and $s_i$, add the value of $v_{j,i}$ to the result register controlled on the values of the variable register and $s_i$. We then apply the inverse of the controlled-$\contains_k$ operations to uncompute the variable ancilla bits. 
\end{proof}

To implement $\calculate_i$, we first define $G_{j,i}$ to be the following circuit: On input we have $\ket{\EffEnc V_{i-1}}$, a counter from $0$ to $L$ using $O(\log L)$ bits, a result bit and $O(\log n)$ workspace ancillas. $G_{j,i}$ first performs a controlled-$\ch_{j,i}$, controlled on the counter being in $0$, on the register containing $\ket{\EffEnc V_{i-1}}$ and $O(\log n)$ of the workspace ancillas. Then, we perform a CNOT onto the result bit, controlled on the ancilla containing $C_j(V_{i-1})$ being $0$. We then run controlled-$\ch_{j,i}^{-1}$, controlled on the counter being in $0$, to uncompute the workspace ancillas. Then, we add $1$ to the counter, controlled on the result bit being $1$. 

The implementation of $\calculate_i$ is now as follows.
\begin{enumerate}
	\item \label{step:calculate1} Run $G_{1,i}\ldots G_{L,i}\ket{\EffEnc V_{i-1}}\ket 0\ket 0$. It is easy to see that this results in $\ket{\EffEnc V_{i-1}}\ket{L+1-j_{\min }}\ket{1-F(x(V_{i-1}))}$, where $j_{\min }$ is the smallest $j$ such that $C_j$ is unsatisfied under $x(V_{i-1})$ if such a clause exist, and $L+1$ otherwise, and $F(x(V))$ is $1$ if $x(V)$ is a satisfying assignment for $F$ and $0$ otherwise.
	\item Run $\select_i$ on $\ket{\EffEnc V_{i-1}}\ket {s_i}\ket{L+1 -j_{\min }}$ and $O(\log n)$ workspace ancillas. This produces $\ket{\EffEnc V_{i-1}}\ket {s_i}\ket{L+1 -j_{\min }}\ket{v_i}$, where $v_i$ is the index of the variable to be added to $V$.
	\item Apply the inverse of step~\ref{step:calculate1}. This results in  $\ket{\EffEnc V_{i-1}}\ket {s_i}\ket{v_i}$ and all other ancilla bits being reset to $0$, as desired. 
\end{enumerate}
This proves Proposition~\ref{prop:phases}(\ref{prop:phasespart1}).\qed

\subsubsection{The merge phase}

The main tool of the merge phase is a reversible operation $\union_{k_1,k_2}$ that calculates the union of two sets $S_1,S_2$ of (known) sizes $k_1,k_2\leq r$,  
\begin{equation}
\union_{k_1,k_2} \ket{\enc S_1}\ket{\enc S_2}\ket{0} = \ket{\enc S_1}\ket{\enc S_2}\ket{\enc S_1\cup S_2}.
\end{equation}

\begin{lemma}\label{lem:union}
 $\union_{k_1,k_2}$ can be  implemented reversibly with  $O(K\log(n/K) + K + \log n)$ ancilla bits and $O(\poly(n))$ gates, where $K=k_1+k_2$. 
\end{lemma}
The basic idea is to simply extract the $j_1$\textsuperscript{th} and $j_2$\textsuperscript{th} elements of $S_1,S_2$ at a time, where $j_1$ and $j_2$ are the current values of two counter registers, appends the smaller of those elements to the output set, and increase either $j_1$ or $j_2$ depending on which one was added. This results in all elements in $S_1\cup S_2$ being added to the output set in increasing order. As such, we first show how to efficiently append an element to a set of known size. 
\begin{lemma}\label{lem:append}
	Let $S=\{v_1,\ldots,v_k\}$ be a set with $0<v_1<\cdots < v_k< 2n$ and let $v\in\{v_k+1,\ldots,2n\}$. Then the operation
	\begin{equation}
		\append_k\ket v\ket{\enc S}\ket{0} = \ket v \ket{\enc S\cup\{v\} }
	\end{equation}
	can be  implemented  reversibly using $O(\log n)$ ancilla bits and trits and $O(\poly(n))$ gates. 
\end{lemma}
\begin{proof}
	The proof is very similar to the proof of Lemma~\ref{lem:extract}, the steps are as follows.
	
	\begin{enumerate}
	\item \label{appendStep1} Use $\extract_k$ to extract the value of the largest element $v_k\in S$ onto a separate register (note that $k$ is known). 
	
	\item\label{appendStep2} Calculate the difference of $v$ with that value and store that difference in a separate register of $O(\log n)$ bits which we call the difference register. 	The ancilla workspace now contains $\ket{v_k}\ket{v-v_k}$.
	
	\item\label{appendStep3}  Introduce a counter from $0$ to $O(\log\log n)$ which we call the length counter, and an additional control bit initially in $0$, and use them to find the number of relevant binary digits in the difference register (\textit{i.e.}, $\lceil \log_2(v-v_k)\rceil+1$) as follows. Starting from the most significant digit of the difference register, in turn do the following for each bit in the difference register: first,  controlled  on the bit being $1$ and end the length counter being $0$, flip the control bit. Then, controlled on the control bit being $1$, add $1$ to the length counter. 
	After doing this for all bits in the difference register, the length counter will be in $\lfloor \log_2(v-v_k)\rfloor+1$ and the control bit in $1$ (since we assume $v>v_k$). Flip the control bit to reset it. The ancilla workspace now contains $\ket{v_k}\ket{v-v_k}\ket{\lfloor\log_2 (v-v_k)\rfloor+1}$. 
	
	\item \label{appendStep4} Use the operation $U_{1,k}$ defined in the proof of Lemma~\ref{lem:extract} to find the position of the last (\textit{i.e.}, $k$\textsuperscript{th}) separation character in $\ket{\enc S}$. 	The ancilla workspace now contains $\ket{v_k}\ket{v-v_k}\ket{\lfloor\log_2 (v-v_k)\rfloor+1}\ket{\pos(S,k+1)}$.
	
	\item\label{appendStep5} Define a family of unitaries $V(j,l)$ which is a sequence of $l$ CNOTs copying the first $l$ bits in the difference register to the binary subspace of the $(j+1),\ldots,(j+l)$\textsuperscript{th} trits  of the set register, followed by flipping the $(j+l+1)$\textsuperscript{st} trit in the set register into a separation character (note that we assume that all these trits are initially in $0$). This operation is only well-defined for $j+l< \lceil (k+1)\log_2(n/(k+1))+4(k+1)\rceil$. For out of bounds values of $j,l$, we define $V(j,l)$ to be the identity.  Define $c-V(j,l)$ to be $V(j,l)$ controlled on the append position counter being $j$ and the length counter being $l$. Apply
	\begin{equation}
		\prod_{j=1}^{\lceil(k+1)\log_2(n/(k+1))+2(k+1)\rceil}\prod_{l=1}^{\lceil \log_2 n\rceil} c_V(j,l).
	\end{equation}
	This appends the value of $v-v_k$ (without leading zeros) and a separation character after the last separation character in the set register.
	
	\item To reset the ancilla workspace, first apply the inverse of step~\ref{appendStep4}, but replacing $U_{1,k}^{-1}$ with $U_{1,k+1}^{-1}$. This resets the ancilla register containing $\ket{\pos(S,k+1)}$. Then apply the inverse of step~\ref{appendStep3} and step~\ref{appendStep2}. This resets the ancilla registers containing $\ket{\lfloor\log_2 (v-v_k)\rfloor+1}$ and $\ket{v-v_k}$, respectively. Finally apply the inverse of step~\ref{appendStep1}, but replacing $\extract_k^{-1}$ with $\extract_{k+1}^{-1}$. This resets the ancilla register containing $v_k$ and thus all remaining ancillas.	\end{enumerate}
	Note that in the last step, the changes from $U_{1,k}^{-1}$ to $U_{1,k+1}^{-1}$ and $\extract_k^{-1}$ to $\extract_{k+1}^{-1}$, respectively, are necessary because after step~\ref{appendStep5}, the set register contains  a set of size $k+1$ and not $k$.   
\end{proof}
\begin{proof}[Proof of Lemma~\ref{lem:union}]
	We introduce the following ancilla registers: two counters from $1$ to $k_1$ and $k_2$, respectively, called the first and second $j$-counter, respectively, two registers of $O(\log n)$ bits called the first and second candidate registers, respectively, $K$ trits called comparison trits and $O(K\log (n/K) + K)$ bits to temporarily store $\ket{\enc S_1\cup S_2}$.
	
	As explained above, the basic idea is to simply extract the elements of $S_1,S_2$ corresponding to the current values of the $j$-counters, then add the smaller of those element to the output set, and increase the corresponding $j$-counter. Additional care however has to be taken to ensure that the algorithm still runs correctly when all elements of one of the sets have been added. As such, the algorithm is as follows. Sequentially, for each $j=0,\ldots,K-1$, do the following:
	
	\begin{enumerate}
	\item \label{step1} For $b=1,2$ in turn, run controlled-$\extract_{k_b}$ on $\ket{\enc S_b}$, the $b$\textsuperscript{th} $j$-counter and the $b$\textsuperscript{th} candidate register, where the operation is controlled on the $(3-b)$\textsuperscript{th} $j$-counter not being in $j-k_1$. Note that since the sum of the two $j$-counters is always equal $j$ at all stages, the latter condition is equivalent to all elements of $S_b$ having already been added. 
	\item \label{step2} For $b=1,2$ in turn, add $n+1$ to the $b$\textsuperscript{th} candidate register, controlled on the $(3-b)$\textsuperscript{th} $j$-counter being in $j-k_1$. This ensures that when all elements of one of the sets have already been added, the value in the candidate register corresponding to the other set is always smaller.
	 \item We compare the values of the two candidate registers and determine the smaller one. To do that, note that it is easy to efficiently and  reversibly implement the minimum finding operation $\Min\ket{v_1}\ket{v_2}\ket 0 = \ket {v_1}\ket{v_2}\ket{m(v_1,v_2)}$, where
\begin{equation}
	m(v_1,v_2) = 
	\begin{cases}
		1 & v_1 < v_2 \\
		2 & v_2 < v_1\\
		0 & v_1 = v_2,
	\end{cases}
\end{equation}
	using only $O(\log\log n))$ ancilla bits. We use the $j$\textsuperscript{th} comparison trit as the third register when calling $\Min$. 
	\item For $b=1,2$ in turn,  
we call controlled-$\append_j$ on the $b$\textsuperscript{th} candidate register and the temporary set register,  	
controlled on the $j$\textsuperscript{th} comparison trit being $b$.
	\item We apply the inverse of the operations in steps \ref{step1}-\ref{step2}. This resets both candidate registers to $0$. 
	\item For $b=1,2$ in turn, we add $1$ to the $b$\textsuperscript{th} $j$-counter controlled on the $j$\textsuperscript{th} comparison trit being in $b$.
	\end{enumerate}
	After doing this for $j=0,\ldots,K-1$, we copy the state of the temporary set register onto the output register and then apply the inverse of the entire circuit up to that point. This resets the comparison trits and both $j$-counters. 	
\end{proof}

We now describe how to implement $\merge_i$. 

The first step is to convert $v_i$ into $\enc \{v_i\}$
using $\append_0$. We then apply a sequence of \emph{merge operations} of sets that  follows the pattern of a binary addition: suppose we expand $i-1$ in binary and add $1$ to it using the standard addition procedure. Then, every carry bit corresponds to merging two sets.
More formally, let $g\geq 0$ be the largest integer such that $2^g$ divides $i$ (g specifies the the first non-zero position in the binary expansion of $i$ from the least significant bit position). Then, the sequence of merge operations is as follows. For $l=0,\ldots,g-1$, we merge the sets $\{v_{i-2^{l+1}+1},\ldots,v_{i-2^l}\}$ and $\{v_{i-2^l+1},\ldots,v_i\}$.

Each  merge of two sets $S_1=\{v_{i-2^{l+1}+1},\ldots,v_{i-2^l}\}$, $S_2=\{v_{i-2^l+1},\ldots,v_i\}$ has two parts. First, we use Lemma~\ref{lem:union} to compute $\ket{\enc S_1\cup S_2}$. Second, we uncompute $\ket{S_1}\ket{S_2}$ by running the inverse of the circuit from the end of the merge phase for $i-2^{l+1}$. Note that since $S_{1,2}$ always contain successive variables up to $v_i$, this is possible. 

We illustrate this procedure for $i=20$. Note that the binary expansion of $i-1$ is $19=16+2+1$, so $V_{19,1} = \{v_1,\ldots,v_{16}\}$, $V_{19,2}=\{v_{17},v_{18}\}$, $V_{19,3}=\{v_{19}\}$. Our aim is to go from
\begin{equation}
	\ket{\EffEnc{V_{19}}}\ket{\enc \{v_{20}\}}
	=
	\ket{\enc\{v_1,\ldots,v_{16}\}}\ket{\enc \{v_{17},v_{18}\}}\ket{\enc \{v_{19}\}}\ket{\enc \{v_{20}\}}
\end{equation} 
to
\begin{equation}	
	\ket{\enc\{v_1,\ldots,v_{16}\}}\ket{\enc \{v_{17},v_{18},v_{19},v_{20}\}} = \ket{\EffEnc V_{20}}.
\end{equation} 
Note that reversibility is preserved since the operation will also involve the $\ket{\vec s}$ register. 
First, we call $\union_{1,1}$ to compute 
\begin{equation}
	\union_{1,1}\ket{\enc \{v_{19}\}}\ket{\enc \{v_{20}\}}\ket 0 = \ket{\enc \{v_{19}\}}\ket{\enc \{v_{20}\}}\ket{\enc \{v_{19},v_{20}\}}.
\end{equation}
 We then uncompute $\ket{\enc \{v_{19}\}}\ket{\enc \{v_{20}\}}$ by running the inverse of the part of the circuit from the end of merge phase of $i=18$. Indeed, that part of the computation mapped $\ket{\enc \{v_1,\ldots,v_{16}\}}\ket{\enc\{v_{17},v_{18}\}}\ket0\ket0$ to $\ket{\enc \{v_1,\ldots,v_{16}\}}\ket{\enc\{v_{17},v_{18}\}}\ket{\enc\{v_{19}\}}\ket{\enc\{v_{20}\}}$. Thus, running the inverse of this part of the circuits results in $\ket{\enc \{v_1,\ldots,v_{16}\}}\ket{\enc\{v_{17},v_{18}\}}\ket{\enc \{v_{19},v_{20}\}}$. Next, we call $\union_{2,2}$, which produces \begin{equation}
\union_{2,2}\ket{\enc\{v_{17},v_{18}\}}\ket{\enc \{v_{19},v_{20}\}}\ket 0 = \ket{\enc\{v_{17},v_{18}\}}\ket{\enc \{v_{19},v_{20}\}}\ket{\enc\{v_{17},v_{18},v_{19},v_{20}\}}.
\end{equation}
 We then uncompute $\ket{\enc\{v_{17},v_{18}\}}\ket{\enc \{v_{19},v_{20}\}}$ by running the inverse of the part of the circuit from the end of the merge phase for $i=16$. This results in 
\begin{equation}
\ket{\enc \{v_1,\ldots,v_{16}\}}\ket{\enc\{v_{17},v_{18},v_{19},v_{20}\}},
\end{equation} which is the desired result.

What remains to be seen is the runtime scaling of $\merge_i$. 
We show now that the runtime of each merge operation (comprising calculating the union of two sets and uncomputing these sets) scales polynomially with the number of elements involved. 
To see this, we first show that the runtime $M_l$ of the operation mapping $\ket{\EffEnc V_{i-2^l}}\ket 0$ to $\ket{\EffEnc V_{i-2^l}}\ket{\enc \{v_{i-2^l+1},\ldots,v_i\}}$ for any $i$ is bounded by $M_l=O^\ast (4^l)$, where we assume that $\ket{\EffEnc V_{i-2^l}}$ has no sets of  $2^l$ elements or less. 
 This can be seen by induction. The claim is trivial for $l=0$. Suppose now that this is true for any $i$ and any $l'<l$. Then, the operation that maps $\ket{\EffEnc V_{i-2^{l+1}}}\ket0$ to  $\ket{\EffEnc V_{i-2^{l+1}}}\ket{\enc \{v_{i-2^{l+1}+1},\ldots,v_i\}}$ comprises two parts. First, we map $\ket{\EffEnc V_{i-2^{l+1}}}\ket0$ to $\ket{\EffEnc V_{i-2^{l+1}}}\ket{\enc\{v_{i-2^{l+1}+1},\ldots,v_{i-2^l}\}} = \ket{\EffEnc V_{i-2^l}}$. This operation takes runtime $M_l$. Next, we map $\ket{\EffEnc V_{i-2^l}}\ket 0$ to $\ket{\EffEnc V_{i-2^l}}\ket{\enc \{v_{i-2^l+1},\ldots,v_i\}}$, which also takes runtime $M_l$. Next, we call $\union_{2^l,2^l}$ to compute 
 \begin{multline}
 \union_{2^l,2^l}\ket{\enc \{v_{i-2^{l+1}+1},\ldots,v_{i-2^l}\}}\ket{\enc \{v_{i-2^{l}+1},\ldots,v_i\}}\ket 0
 =\\
 \ket{\enc \{v_{i-2^{l+1}+1},\ldots,v_{i-2^l}\}}\ket{\enc \{v_{i-2^{l}+1},\ldots,v_i\}}\ket{\enc \{v_{i-2^{l+1}+1},\ldots,v_i\}}.
 \end{multline}
 The runtime of this call can be bounded by a polynomial  $p(n)$ that is independent of $l$ or $i$. Finally, to uncompute $\ket{\enc \{v_{i-2^{l+1}+1},\ldots,v_{i-2^l}\}}\ket{\enc \{v_{i-2^{l}+1},\ldots,v_i\}}$, we apply of the inverse of the two operations before the call to the union, each taking runtime $M_l$. This implies $M_{l+1}\leq 4M_l + p(n)$, which clearly gives $M_l = O^\ast(4^l)$, as claimed. In particular, this implies that merging $\{v_{i-2^{l+1}+1},\ldots,v_{i-2^l}\}$ and $\{v_{i-2^{l}+1},\ldots,v_{i}\}$, comprising of first calculating their union and then uncomputing the original sets, takes runtime at most $O^\ast(4^l)$. 
 
 Since $i\leq r$, we need to do this operation at most once for each $l\in\{0,\ldots,\lceil \log_2 r\rceil\}$. Thus, the runtime of $\merge_i$ is at most 
 \begin{equation}
 	O^\ast\left( \sum_{l=0}^{\lceil \log_2 r\rceil} 4^l \right) = O^\ast(\poly(r)) = O(\poly(n)).
 \end{equation}
 
 This proves Proposition~\ref{prop:phases}(\ref{prop:phasespart2}). \qed
 
 \subsubsection{Quantum algorithm for PBS}
  
 To summarize, we have proven the following

\begin{proposition}[Classical reversible circuit for $\textsc{QBall}_1$] \label{prop:classicalQBALL1}  			 The map $\ket{\vec s}\ket 0\mapsto \ket{\vec s}\ket{\EffEnc V(\vec s)}$ can be implemented reversibly using $O(r\log (n/r) + r + \log n)$ ancillas and $O(\poly(n))$ gates.
		
\end{proposition}

For completeness, we also show how to (classically) reversibly imeplement $\qball_2$.

\begin{proposition}[Classical reversible circuit for $\textsc{QBall}_2$]\label{prop:classicalQBALL2}
	The map $\ket{\EffEnc V(\vec s)} \ket 0 \mapsto \ket{\EffEnc V }\ket{F(x(V))}$ can be implemented reversibly using $O(\log n)$ ancillas and $O(\poly(n))$ gates. 
\end{proposition}
\begin{proof}
	The algorithm is similar to $\calculate_i$. Let $G_{j,i}$ be defined as in the implementation of $\calculate_i$. We can assume that $V=V(\vec s)$ has exactly $r$ elements. The algorithm is as follows.
	\begin{enumerate}
	\item \label{step:qball2_1} Run $G_{1,r+1}\ldots G_{L,r+1}\ket{\EffEnc V_{r}}\ket 0\ket 0$. This results in $\ket{\EffEnc V_{r}}\ket{L+1-j_{\min }}\ket{1-F(x(V_{i-1}))}$, where $j_{\min }$ is the smallest $j$ such that $C_j$ is unsatisfied under $x(V_{r})$ if such a clause exist, and $L+1$ otherwise.
	\item Use a CNOT to copy the last bit onto the output bit, invert it, and apply the inverse of step~\ref{step:qball2_1} to reset the ancillas. \qedhere
\end{enumerate}
\end{proof}

As mentioned above, this can now easily be turned into a quantum algorithm for PBS.

\begin{theorem}[Quantum algorithm for PBS]
	There exists a quantum algorithm that solves $r$-PBS in runtime $O^\ast(3^{r/2})$, using at most $O(r\log(n/r)+r+\log n)$ qubits.
\end{theorem}
 \begin{proof}
 First, quantize the classical reversible circuits of Proposition~\ref{prop:classicalQBALL1}--\ref{prop:classicalQBALL2} by turning each (reversible) classical gate into its quantum equivalent. Then, We initialise the $\ket{\vec s}$ register into 
 \begin{equation}
 	\frac1{\sqrt{3^r}}\sum_{s_1,\ldots,s_r=1}^3 \ket{s_1,\ldots,s_r}
 \end{equation}
 and apply the (quantum) $\textsc{QBall}_{1}$ followed by $\textsc{QBall}_2$. The latter produces the state
 \begin{equation}
 	\frac1{\sqrt {3^r}}\sum_{s_1,\ldots,s_r=1}^3 \ket{s_1,\ldots,s_r}\ket{V(\vec s)}\ket{F(V(\vec s))}.
\end{equation}  
The last step is to run amplitude amplification (or alternatively, fixed point search \cite{2014_Yoder}) to increase the overlap with $1$ on the last qubit. This uses at most $O(\sqrt{3^r})$ repetitions of $\textsc{QBall}_{1,2}$. The overall runtime of the quantum algorithm is therefore $O^\ast(3^{r/2})$. 
\end{proof}

\subsection{Hybrid algorithm and runtime analysis}
\label{sm:runtimehybrid}

In this section, we provide the details of the full quantum-enhanced algorithm to solve 3SAT using a small quantum device. 
In Section~\ref{SM0} we provide the basic ideas and the broad overview. The subsequent sections provide all the details of the analysis for completeness. Specifically, in Section~\ref{SM1}, we first look at general runtime properties of the space splitting algorithm which reduces 3SAT to PBS. In Section~\ref{SM2}, we summarise the classical algorithm to solve PBS from \cite{2011_Moser}, which is more efficient than $\promiseball$. In Section~\ref{SM3}, we show how that algorithm can be quantum enhanced using $\qball$ and derive it's runtime. The full runtime to solve 3SAT is finally derived in Section~\ref{SM4}.

\subsubsection{Key ideas on the hybrid algorithm}
\label{SM0}
As clarified in the main text, the results of \cite{2002_Dantsin} have shown that an algorithmic speed-up for the Promise Ball problem can lead to a faster algorithm for SAT solving. This connection is given quantitatively with Eq. (\ref{eq} ) in Section~\ref{SM1} below. 
Specifically, any algorithm which solves Promise Ball faster than the algorithm of Sch\"{o}ning (recall, Sc\"{o}ning sampling can equivalently be used to solve Promise Ball) will also outperform Sch\"{o}ning's algorithm 
on the overarching SAT problem itself.
Thus the focus of our approach is in providing speed-ups for the Promise Ball problem. 

In the previous section of this Supplemental Materials, we have provided the details of the quantum algorithm which space-efficiently solves the Promise Ball problem, provided sufficiently many qubits are available. 
However, in the space-restricted scenario, we have to resort to other methods.
Specifically, we focus on a more advanced algorithm, a de-randomization of the algorithm of Sch\"{o}ning, for classical Promise Ball solving provided in \cite{2011_Moser}. This algorithm has two features critical for our purposes. First, even without quantum enhancements, this algorithm is (essentially) as fast as the original algorithm of Sch\"{o}ning. This ensures that no thresholds will emerge, relative to the algorithm of Sch\"{o}ning, as any speed-up will imply beating the original Sch\"{o}ning run-time \footnote{More specifically, since we quantum-enhance the algorithm of \cite{2011_Moser}, our approach circumvents a threshold phenomenon relative to this algorithm. Since the algorithm of \cite{2011_Moser} is referred to as a de-randomization of the algorithm of {Sch\"{o}ning}, with matching run-time, we for simplicity talk about avoiding the threshold relative to the algorithm of {Sch\"{o}ning} itself.}.
Second, it is a divide-and-conquer algorithm: it works by calling itself on an instance smaller relative to the all the important parameters. This is in general a non-trivial demand: the Promise Ball problem comes with at least two relevant parameters: radius $r$ and the number of variables $n$. In terms of classical time-complexity, the run-time of the algorithm for Promise Ball is exponential only with respect to $r$, and not $n$ \footnote{For the interested reader we point out that the fine-graining of complexity of a given problem with respect to multiple parameters of a given instance, as is the case here, is studied by so-called parametrized complexity theory.}, so recursing over $r$ makes sense in the classical regime.
However, since we are dealing with a (very) limited-size machine, reducing the instance size for Promise Ball only relative to the radius $r$ would not in general suffice. One would still seemingly need to represent the entire formula on the quantum device, which is impossible since we assume we have (significantly) fewer qubits than $n$. Luckily, we have shown that we do not need to carry the representation of the formula as input. This is critical for the standard hybrid approach, where the quantum routine is invoked when the instance is small enough, to be applicable. 

A final contribution of this Section is the overall run-time analysis, provided in Sections \ref{SM3} and \ref{SM4}. The key technical subtleties of this analysis are two-fold. First, a quantum machine which can handle Promise Ball over $n$ variables up to radius $r$ requires somewhat more than $r$ (qu)bits. 
On the other hand, the expressions quantifying run-time incorporate the quantities related to the radius $r$ that can be handled. 
This implies that any expression which quantifies the total run-time relative to a size of the quantum device must include the explicit functional relationship between $r,$ $n$  and the number of qubits we need to solve Promise Ball with  parameters $(n,r)$ on a quantum computer.
The second issue has to do with details of the faster algorithm in \cite{2011_Moser}, where the recursive calls reduce the instance sizes with respect to non-unit steps of $\Delta>1,$ where $\Delta$ influences the efficiency of the overall algorithm. Due to this, the precise analysis of the achieved speed-up via the standard hybrid approach is slightly more involved than for the basic, slower, Promise Ball algorithm presented in the main text. 
 These technical points are elaborated in detail in section \ref{SM3}.

\subsubsection{runtime properties of the algorithm of \cite{2002_Dantsin}}
\label{SM1}
In \cite{2002_Dantsin}, it was shown that using the space splitting algorithm (which reduces 3SAT to PBS), 3SAT  over $n$ variables can be solved in time
\EQ{
T(n) =  \underbrace{q_{d}(n) (2^{3n/d}\! +\! 2^{(1-h(\rho))n})}_{\textup{Cover}\ \textup{Set}\ \textup{preparation}}\! +\! \underbrace{q_{d}(n) 2^{(1-h(\rho))n}}_{\textup{Cover}\ \textup{Set}\ \textup{size}}\! \times \underbrace{T_{2}(n,\rho)}_{\textup{PBS}\ \textup{cost}}, \label{eq}
}
where $\rho$ (the fraction specifying the radius of the balls via $r = \rho n$) and $d$ are parameters which can be optimized, $q_{d}$ is a polynomial depending on $d$, $h(\rho) = - (\rho \log_2(\rho) + (1-\rho) \log_2(1-\rho))$ is the binary entropy function, and $T_{2}(n,\rho)$ is the runtime of the algorithm used to solve PBS with $n$ variables and radius $r=\rho n$.

 Sch\"{o}ning's algorithm can also be understood as a PBS solver. It can be shown that  Sch\"oning sampling, starting from a center $\mathbf{x}$, which is at Hamming distance $r$ from a satisfying assignment, produces a satisfying assignment with probability at least $2^{-r}$ \cite{2011_Moser}. By iteration, we obtain a PBS solver with runtime $O^{\ast}(2^r)$.  
  
Since $d$ can be chosen large enough such that the dominating term of $T(n)$ is $2^{(1-h(\rho))n} \times T_{2}(n,\rho)$ \footnote{Since $d$ appears in $2^{3n/d}$, to make sure this contribution is not a dominating term in the overall complexity, setting $d=15$ will suffice as $2^{n/5} \leq 2^{\gamma_q n}$)}
 to quantify the improvement given a quantum-enhanced subroutine, it suffices to optimize this term with respect to $\rho$. If (the dominating part of) $T_2$ is of the form $T_2(n,\rho) = O^\ast(2^{\zeta \rho n})$, the optimum is obtained by just minimizing $1-h(\rho) + \zeta \rho$. Important values of $\zeta$ are $\zeta=\log_2{(3)}$ and $\zeta=1$, corresponding to PBS solved by the deterministic  algorithm of \cite{2002_Dantsin} and the basic (and also fast deterministic) Schoning's algorithm, attaining optima at values $\rho = 1/4 $, and $\rho = 1/3$, respectively.
For completeness, the overall effective values $\gamma,$ for the overall algorithm using PBS routines as listed are approximately $0.58$ and $0.41,$ where the latter matches the runtime of  Sch\"{o}ning's (original) algorithm.

\subsubsection{Fast deterministic classical PBS solver}
\label{SM2}
 In  \cite{2011_Moser}, an improved deterministic classical PBS solver was introduced, with run time in 
 $O^{\ast}( (2+\epsilon)^r),$ where $\epsilon$ depends on tunable protocol parameters. We call this algorithm $\textsc{FastBall}$. These parameters can be chosen such that $\epsilon$ is arbitrarily small or even decaying in $r$, as explained later.
 As mentioned, the key idea of this algorithm is to also split the space of choices, selecting which literal will be flipped in the recursive call, into covering balls. 
 For the convenience of the reader, here we present the details of $\textsc{FastBall}$, adapted from \cite{2011_Moser}.
Let $t \in \mathbbmss{N}$ be a parameter (influencing $\epsilon$ in the overall runtime), and $\mathcal{C} \subseteq \{1, \ldots, k \}^t$ be a $k-$arry covering code with radius $t/k$ (where $k$ specifies the clause upper bound in the $k$SAT problem to be solved).
Since $t$ and $k$ are constants, the optimal code can be found in constant time.
%Then we have the following algorithm, 
The key steps and performance aspects of the algorithm are given next abbreviated and adapted from \cite{2011_Moser}.
\begin{figure}[H]
\begin{algorithmic}[1]
\Procedure{\textsc{FastBall}}{$F, \mathbf{x}, {r},\mathcal{C}$}
  \State If $\mathbf{x}$ satisfies $F$
  \State \indent return $\mathbf{x}$
  \State else if $r = 0$
  \State \indent return FALSE 
  \State else
  \State \indent $G \leftarrow $ a maximal set of pairwise disjoint $k$-clauses of $F$ unsatisfied by $\textbf{x}$
   \State\label{case1}  \indent if $|G| < t$,  \Comment{Case 1}
   \State \indent \indent for each assignment $\beta$ of variables in $G$
   \State \indent \indent \indent call \textsc{PromiseBall}$(F_{| \beta},\mathbf{x},r),$ 
 \State \label{case2} \indent else if $|G| \geq t$  \Comment{Case 2}
\State \indent \indent  $H \leftarrow \{C_1,\dots,C_t\} \subseteq G$
  \State \indent \indent for each {$w \in \mathcal{C}$}
\State\label{FastBallRecursionStep}  \indent \indent \indent  call $\textsc{FastBall}(F, \mathbf{x}[H,w], r- \Delta, \mathcal{C})$  \Comment $\Delta \mathop{:}=  t-2t/k $
\EndProcedure
\end{algorithmic}
\end{figure}
The parameter $\Delta$ influences $\epsilon$, as will be clarified presently. 
The expression $\mathbf{x}[H,w]$ corresponds to a modified assignment, where the variables selected by the code-word $w$ from the subset of variables occurring in the subset of clauses $H$ have been flipped.

The objective of the algorithm is to achieve the run time of essentially $O^\ast((k-1)^r),$ which would be the runtime achieved by the randomized algorithm of Sch\"{o}ning, used as a PBS solver.
We now briefly discuss the runtime of this algorithm.

In the case the problem is such that the algorithm always encounters Case 1 in line~\ref{case1}, the authors in \cite{2011_Moser} show that
each $F_{| \beta}$ has no clauses of size larger than $k-1$. In this case,  Proposition 7 in \cite{2011_Moser} shows that  \textsf{PromiseBall} can solve the problem in $O^\ast((k-1)^r)$, yielding the overall runtime is $ O^{\ast}(2^{kt}(k-1)^r)$. Since $2^{kt}$ is a constant, this is $O^\ast((k-1)^r),$ achieving the objective.

The more complex case involves occurrences of Case 2 (in line~\ref{case2}). Here, recurrence calls occur, which may encounter Case 1 deeper in the tree, or not. The slowest case occurs when we remain in Case 2 throughout recurrence calls.
If we set $\Delta =  t-2t/k $, it was shown in \cite{2011_Moser}  that
the runtime of $\textsc{FastBall}$, which, up to polynomial factors, is the number of leaves in the recursion tree, is $O^\ast\left((t^2 (k-1)^{\Delta})^{r/\Delta}\right) = O^\ast\left((t^{2/\Delta} (k-1))^r\right)$. Since $t^{2/\Delta}$ goes to 1 as $t$ grows, for any $\epsilon>0$ we can choose $t=t(\epsilon)$ such that $(t^{2/\Delta} (k-1))^r \leq (k-1 + \epsilon)^{r},$ which is the main result. 
Note that $t$ can also be chosen to be a very slowly growing function of $n$, $t = \log_k \log_2 (n)$, which guarantees that the runtime of the algorithm, intuitively, approaches the expression of the form $O^{\ast}((k-1)^r)$ as $n$ grows (see \footnote{http://users-cs.au.dk/dscheder/SAT2012/searchball.pdf}).

\subsubsection{Quantum speedup of $\textsc{FastBall}$ with a  small quantum device}\label{SM3}
The algorithm \textsc{FastBall} recursively calls itself on smaller instances, where the value of $r$ is reduced in steps of $\Delta$ (where $\Delta = t/3$ for 3SAT), and $\Delta$ is a function which depends on $\epsilon$ alone -- hence if $\epsilon$ is fixed, $\Delta$ is a constant. 

At each recurrence step, the algorithm first checks whether a criterion which would ensure that \textsc{PromiseBall} would terminate in $O^\ast((k-1)^{r})$ ($O^\ast(2^{r})$  for 3SAT) steps, is satisfied (Case 1). In that case, the algorithm runs $\textsc{PromiseBall}$, the quantum enhancement of which was investigated in the main body of the paper. 
The faster runtime in that case is ensured by the property that the formula in question actually has at most $(k-1)$ (2 for 3SAT) variables per unsatisfied clause. In that case, a quantum enhancement is achieved by the standard hybrid approach, where  \textsc{QBall}  is run as soon as the instance becomes small enough. Note that since the relevant formula has only $k-1$ variables per clause, $\textsc{QBall}$ can be adapted in this case to yield a runtime of $O^\ast((k-1)^{ r/2})$ if $ r$ is small enough. 
Then, we obtain an interpolated time between Sch\"{o}ning-level performance, and something quadratically faster. We do not need to delve on further details, since this is not the worst-case performance of the algorithm, which occurs if Case 2 persists.
We cannot beforehand know at which step, if at all, the criteria  for Case 1 will be satisfied, so  we call $\textsc{QBall}$ in Case 2 as well, as soon as the recursive step in line~\ref{FastBallRecursionStep} calls an instance with sufficiently small $r$. This $\qball$-enhanced version of $\fastball$ is a hybrid algorithm which we call $\textsc{QFastBall}$.

We now estimate the runtime of $\textsc{QFastBall}$. Recall that we assume a quantum computer with $M=cn$ qubits, where $c\in(0,1)$ is an arbitrary constant. To obtain the runtime of $\textsc{QFastBall}$, we first determine the largest value of $\tilde r=\tilde r(c,n)$ such that such a quantum computer can solve PBS of radius $\tilde r$.  

Recall that $\textsc{QBall}$ requires $O( r\log(n/ r)+ r + \log n)$ qubits to solve PBS with $n$ variables and radius $ r$. Suppose the exact scaling is of this number of qubits is $Ar\ln(n/r) + Br + O(\log n)$ for constants $A,B>0$ \footnote{The implementation given in Section~\ref{sm:qPBS} yields $A=10$ and $B=50$ using a straightforward of encoding each trit using two qubits. We remark however that these numbers can be improved significantly.}. Let $\beta(c)>0$ be such that $A\beta(c)\ln(1/\beta(c)) + B\beta(c)= c$. Then, the quantum device can solve PBS for all $r\leq \tilde r=\beta(c)n - O(\log n)$. For completeness, it can be shown that 
$
	\beta(c) = -c/(AW(-ce^{-B/A}/A))
$,
where $W$ is the principal branch of the Lambert $W$ function (note that $W(x)<0$ for $x\in(-1/e,0)$). We will not need the precise form of $\beta(c)$, it is easy to see, however, that for small values of $c$, $\beta(c) =\Theta(c/\log(1/ c))$. 

Suppose now that $\textsc{QFastBall}$ calls $\textsc{QBall}$ with radius $\rcall$. Clearly, $\tilde r \geq\rcall > \tilde r - \Delta$ for each call to $\qball$. Note that the value of $\rcall$ could be different for each call to $\qball$ depending on whether we call it deep in $\promiseball$ (Case 1) or in Case 2. The number of calls to $\qball$ is however at most $O((2t^{2/\Delta})^{r-(\tilde r - \Delta)})$, since the number of leaves in the recursion tree, which has depth at most $r-(\tilde r - \Delta)$,  before $\qball$ is called, is bounded by this quantity. The runtime of $\textsc{QFastBall}$ is therefore given by product of this quantity and the runtime $O^\ast(3^{\rcall/2})$ of $\qball$, \textit{i.e.}, at most
$
	O^\ast((2t^{2/\Delta})^{r-(\tilde r - \Delta)} \cdot 3^{\tilde r/2})
	=
	O^\ast ((2t^{2/\Delta})^{r-\tilde r} \cdot 3^{\tilde r/2})
$, 
since $\Delta$ and $t$ are constants. Plugging in the expression for $\tilde r$ and noting that the $O(\log n)$ contribution is absorbed by the $O^\ast$ notation, we obtain a runtime of
\begin{align}
	T_{2,\textsc{QFastBall}}(\rho,n)&=O^\ast\left((2+\epsilon)^{\rho n} \left( \frac{\sqrt 3}{2+\epsilon}\right)^{\beta(c) n}\right) \\
	&=
	O^\ast\left((2+\epsilon)^{\rho n} 2^{-f(c) n} \right) \label{eq:runtimeqfastball}
\end{align}
for solving PBS with $n$ variables and radius $r$ using $\textsc{QFastBall}$ with a quantum device of $M=cn$ qubits, where as before $\epsilon$ can be made arbitrarily small, and $f(c) = (1-\log_2\sqrt 3)\beta(c)\approx 0.21\beta(c)$, where in the last step, we bounded $\sqrt 3/(2+\epsilon)< \sqrt 3/2$.  %

 \subsubsection{Total runtime for 3SAT}
\label{SM4}

We now estimate the runtime of the entire algorithm for solving 3SAT using the space splitting algorithm in combination with $\textsc{QFastBall}$. Substituting \eqref{eq:runtimeqfastball} into \eqref{eq}, and recalling that we can choose $d=O(1)$ such that the second term in \eqref{eq} is the dominating term of $T(n)$, we obtain
\begin{equation}\label{eq:ttot}
	T(n) = O^\ast\left( 2^{(1-h(\rho))n}(2+\epsilon)^{\rho n} 2^{-f(c)n} \right).
\end{equation}
Note that the last factor in \eqref{eq:ttot}, induced by the quantum enhancement, is independent of $\rho$, and that the first two factors together constitute the runtime of the space splitting algorithm using just $\fastball$. The optimal value of $\rho$ for \eqref{eq:ttot} is therefore the same as the optimal value of $\rho$ for just using $\fastball$, which (up to corrections in $\epsilon$) is $\rho=1/3$. Using this value, we obtain a total runtime of
\begin{equation}
	T(n) = O^\ast \left( 2^{(\gamma_0+\varepsilon-f(c))n}\right),
\end{equation}
where $\varepsilon= \log_2(1+\epsilon/2)/3$ can be made arbitrarily small, as stated in the main body of the paper.
\end{document}